\documentclass[%
 preprint, 
 amsmath,amssymb,
 aps, physrev,
pra,
]{revtex4-2}

\usepackage{mathpazo}
\usepackage{amsfonts}
\usepackage{amsmath}
\usepackage{bm}
\usepackage{graphicx}
\usepackage{latexsym}

\usepackage[margin=1.05in]{geometry}

\usepackage[T1]{fontenc}
\usepackage{times}
\usepackage{color,graphicx}
\usepackage{array}
\usepackage{enumerate}
\usepackage{amsmath}
\usepackage{amssymb}
\usepackage{amsthm}
\usepackage{bbm} 
\usepackage{pgfplots}
\usepackage{pgf}
\usepackage{graphicx}
\usepackage{tikz}
\usetikzlibrary{patterns}
\usetikzlibrary{arrows.meta}
\usepgfplotslibrary{patchplots} 
\usetikzlibrary{pgfplots.patchplots} 
\pgfplotsset{width=9cm,compat=1.5.1}
\usepackage{diagbox}

\definecolor{myblue}{RGB}{0,0,128}
\definecolor{myblue2}{RGB}{0,32,96}
\definecolor{myblue3}{RGB}{0,64,64}
\definecolor{myblue4}{RGB}{0,96,32}
\definecolor{myblue5}{RGB}{0,128,0}
\definecolor{myblue6}{RGB}{32,96,0}
\definecolor{myblue7}{RGB}{64,64,0}
\definecolor{myblue8}{RGB}{96,32,0}
\definecolor{myblue9}{RGB}{128,0,0}

\usepackage{amsthm}

\newtheorem{theorem}{Theorem}

\newtheorem{lemma}[theorem]{Lemma}

\usepackage{xcolor}
\usepackage{makeidx}
\usepackage[colorlinks=true,linkcolor=blue,anchorcolor=blue,citecolor=red,urlcolor=magenta]{hyperref}
\usepackage{caption}
\usepackage{subcaption}

\definecolor{ballblue}{rgb}{0.13, 0.67, 0.8}

\newcommand{\tr}{\operatorname{Tr}}

\newcommand{\ket}[1]{| #1 \rangle}

\newcommand\supp{\mathop{\rm supp}\nolimits}

\usepackage{scalerel}

\def\I{\mathbb{1}} 
 
\def\C{\mathbb{C}}

\def\K{\mathcal{K}}

\def\E{\mathcal{E}}

\def\V{\mathcal{V}}
\def\W{\mathcal{W}}
\def\X{\mathcal{X}}

\def\l{\ell}

\def\0{\mathbf{0}}

\def\spn{\text{span}}
\def\supp{\text{supp}}
\def\ker{\operatorname{ker}}
\def\ima{\operatorname{Im}}

\def\Pos{\operatorname{Pos}}

\begin{document}

\title{Optimal discrimination of quantum sequences}

\author{Tathagata Gupta}
\affiliation{Physics and Applied Mathematics Unit, Indian Statistical Institute, 203 B. T. Road, Kolkata 700108, India}
\email{tathagatagupta@gmail.com} 

\author{Shayeef Murshid}
\affiliation{Electronics and Communication Sciences Unit, Indian Statistical Institute, 203 B. T. Road, Kolkata 700108, India}
\email{shayeef.murshid91@gmail.com} 

\author{Vincent Russo}
\affiliation{Unitary Fund, 315 Montogomery St, Fl 10 San Francisco, California 94104, USA} 
\email{vincent@unitary.fund}

\author{Somshubhro Bandyopadhyay}
\affiliation{Department of Physical Sciences, Bose Institute, EN 80, Bidhannagar, Kolkata 700091, India}
\email{som@jcbose.ac.in}


\begin{abstract}
A key concept of quantum information theory is that accessing information encoded in a quantum system requires us to discriminate between several possible states the system could be in. A natural generalization of this problem, namely, quantum sequence discrimination, appears in various quantum information processing tasks, the objective being to determine the state of a finite sequence of quantum states. Since such a sequence is a composite quantum system, the fundamental question is whether an optimal measurement is local, i.e., comprising measurements on the individual members, or collective, i.e. requiring joint measurement(s). In some known instances of this problem, the optimal measurement is local, whereas in others, it is collective. But, so far, a definite prescription based solely on the problem description has been lacking. In this paper, we prove that if the members of a given sequence are secretly and independently drawn from an ensemble or even from different ensembles, the optimum success probability is achievable by fixed local measurements on the individual members of the sequence, and no collective measurement is necessary. This holds for both minimum-error and unambiguous state discrimination paradigms. 

\end{abstract}

\maketitle

\section{Introduction}
One of the characteristic features of quantum theory is that composite quantum systems can possess nonlocal properties. This is often associated with entangled systems violating a Bell inequality \cite{bell1966problem,brunner2014bell}. However, an unentangled system, whose parts may not have interacted in the past, can also exhibit nonlocality, conceptually different from the Bell type \cite{peres1991optimal, bennett1999quantum, wootters2006distinguishing,chitambar2013revisiting,halder2019strong}. Specifically, when determining the state of an unentangled system, known to be in one of several possible states, a joint measurement of the whole system is sometimes necessary. Thus a fundamental question in quantum information is when does optimal extraction of information from a composite system demands access to the whole and when does it not. 

This question frequently arises in many quantum information theoretic protocols of practical importance, such as quantum change-point detection \cite{sentis2016quantum,sentis2017exact,mohan2023generalized}, multiple-copy state discrimination \cite{higgins2011multiple}, and quantum key distribution \cite{cerutti2024optimal}. In all these scenarios, essentially, the objective is to determine the state of an unknown quantum sequence $\left(\rho_{1},\dots,\rho_{k}\right)$ whose members $\rho_{i}$ are drawn from ensembles of which we have complete knowledge. For a given sequence of finite length, this boils down to the problem of sequence discrimination, a natural generalization of the well-studied quantum state discrimination problem \cite{bae2015quantum,barnett2009quantum, chefles2000quantum, bergou200411}, where one aims to discriminate between the possible states a quantum system could be in. Since the state of a sequence $\left(\rho_{1},\dots,\rho_{k}\right)$ is of the form $\rho_{1}\otimes\cdots\otimes\rho_{k}$, the question is whether the optimal measurement, i.e., optimal according to some well-defined figure of merit, is local, i.e., comprising measurements on the individual members, or collective, i.e., requiring joint measurement(s). Moreover, even when local access is sufficient, it is necessary to know whether coordinated or adaptive strategies provide any advantage over fixed strategies that do not involve adaptation based on the outcomes of already performed measurements. In fact, we know of instances where the optimal measurement is local and fixed \cite{gupta2024unambiguous}, local and adaptive \cite{acin2005multiple}, and collective \cite{calsamiglia2010local}; in particular, in the quantum change-point problem, it is collective \cite{sentis2016quantum}, whereas in multiple-copy discrimination, depending on the problem specification, it could be either of the two \cite{acin2005multiple,calsamiglia2010local}.

Unfortunately, given a sequence discrimination problem, we do not have any definite characterization or prescription that could, at the very least, tell us about the nature of the optimal measurement. As it is, discriminating between states of a composite quantum system is known to be challenging \cite{wootters2006distinguishing, chitambar2013revisiting, halder2019strong} and for sequences more so as they could be of any finite length and also come in different varieties: they could be independent and identically distributed (i.i.d.), where the members of a sequence are independent (i.e., selection of an individual member does not depend on what states have already appeared in the earlier positions) and identically distributed (i.e., the members are drawn from the same ensemble), or non-i.i.d., where members are either drawn independently but from different ensembles or not independently at all.

In this paper, we take a significant step towards addressing whether an optimal measurement is local or collective only from the problem description. We prove that if the members of a sequence are secretly and independently selected either from the same ensemble or even from different ensembles, the optimal measurement is local; in particular, it constitutes fixed measurements on the individual elements. The result holds for both minimum-error and unambiguous state discrimination paradigms. For the latter, however, additional assumptions are required for a nontrivial result since unambiguous discrimination applies to states that satisfy certain conditions \cite{chefles1998unambiguous,feng2004unambiguous}. 

Note that as long as the members of a given sequence are selected independently, \textit{irrespective of the probability distribution associated with the selection being identical or different}, the optimal measurement to determine its state is always local and fixed. Therefore, if the independence condition does not hold, all types of measurements mentioned earlier, are possible. Future works, therefore, only need to address the optimality question in the non-independent scenario. 

\section{Problem statement and main result}
Suppose we are given an unknown sequence of length $k\in\mathbb{N}$ and wish to determine its state as well as possible. The members of the given sequence are secretly and independently drawn from potentially different but known ensembles. The ensemble from which the $i$th member is drawn is labeled as $\mathcal{E}^{i}$, where \begin{equation}
\mathcal{E}^{i}=\left\{ \left(\eta_{j}^{i},\rho_{j}^{i}\right):j=1,\dots,\l_{i}\right\}\end{equation} with $\rho_{j}^{i}$ being density operators on $\mathbb{C}^{d_{i}}$ for $d_{i}\geqslant2$. Thus, for a sequence of length $k$, we have $k$ such ensembles $\mathcal{E}^{1},\dots,\mathcal{E}^{k}$. In this way, we can account for all possibilities: (i) every member is drawn from a different ensemble; some are drawn from the same ensemble, in which case, though the labeling of ensembles is different, they are, in fact, identical; and (iii) every member is drawn from the same ensemble, so here no indexing of ensembles would be necessary. Note that the last one corresponds to i.i.d. sequences, whereas every other case corresponds to independent but not identically distributed sequences.

Thus, given a sequence of length $k$, we only know that the $i$th member is drawn from $\mathcal{E}^{i}$ and $\eta_{j}^{i}$ is the prior probability of it being $\rho_{j}^{i}$. With this, the ensemble of all possible sequences of length $k$ is given by \begin{equation} \label{eq:seqens}
\mathcal{E}_{k}=\left\{ \left(\eta_{x_{1}}^{1}\cdots \eta_{x_{k}}^{k},\rho_{x_{1}}^{1}\otimes\cdots\otimes\rho_{x_{k}}^{k}\right)\right\},\end{equation} where each $x_{i}\in\left[\l_{i}\right]$ for all $i\in\left[k\right]$. A simple counting argument shows that $\mathcal{E}_{k}$ contains $\l=\prod_{i=1}^{k}\l_{i}$ sequences. The objective is to determine whether the optimal measurement discriminating between the elements of $\mathcal{E}_{k}$ is collective or local. 

The optimality of a measurement subject to a given set of states depends on the choice of measurement strategy. Here, we consider both minimum-error and unambiguous state discrimination paradigms. The former minimizes the average error and applies to any set of states. The corresponding measure is the success probability, the maximum probability that the unknown state is correctly determined \cite{holevo1973statistical,helstrom1969quantum}. The latter \cite{ivanovic1987differentiate,dieks1988overlap,peres1988differentiate}, however, applies to sets of states that satisfy a specific condition, and if this condition is met, this approach correctly determines the unknown state with a nonzero probability; for example, a set of pure states can be unambiguously discriminated if and only if they are linearly independent \cite{chefles1998unambiguous}. Closed-form solutions for both are known in the two-state case \cite{helstrom1969quantum,jaeger1995optimal} and in specific instances where the states satisfy certain symmetry properties \cite{sun2001optimum}. For generic ensembles, finding solutions is generally hard; however, semidefinite programs exist \cite{eldar2003designing,eldar2003semidefinite}, and the optimum success probability can be obtained as the output of such programs. 

The main result of this paper is the following theorem.
\begin{theorem}\label{thm:main}
Let $p\left(\mathcal{E}\right)$ denote the optimum probability for minimum-error or unambiguous discrimination between the elements of an ensemble $\mathcal{E}$. Then, \begin{equation}\label{eq:main} p\left(\mathcal{E}_{k}\right) =\prod_{i=1}^{k}p\left(\mathcal{E}^{i}\right). \end{equation} 
\end{theorem}
Thus the optimal measurement for discriminating between the elements of $\mathcal{E}_{k}$ is local and comprises of individual measurements on the members of the sequence. In particular, for the $i$th member, the local measurement corresponds to the optimal measurement that discriminates between the elements of the ensemble $\mathcal{E}^{i}$. 

Since minimum-error discrimination is possible for any set of states, one would always have a nonzero success probability. Specifically for sequences, this means that for any set of ensembles $\{\E^i\}$, we have $p(\E^i)>0$ for all $i$ and as a result $p(\E_k)>0$ in Eq. \eqref{eq:main}. However, for unambiguous discrimination this will not be the case in general. Therefore, our result shows that for unambiguous discrimination of a set of quantum sequences, a nontrivial result, i.e., $p\left(\mathcal{E}_{k}\right)>0$, is obtained if and only if for each ensemble $\mathcal{E}^{i}, p(\E^i)>0$, that is, the corresponding ensembles can be unambiguously discriminated.

It is useful to compare our main result with a couple of previous ones obtained under additional assumptions. The authors in \cite{wallden2014} studied minimum-cost measurements for quantum states, which is a generalization of minimum-error discrimination. The average cost of the measurement $M=\{M_i\}_{i=1}^N$ for the ensemble of states $\{(q_i,\sigma_i)\}_i, \; i \in [N]$ with respect to the cost matrix $C=[C_{ij}]$ is denoted $\tilde{C}(M)$ and defined as \begin{equation}
\tilde{C}(M)=\sum_{ij}q_iC_{ij}\tr(\sigma_iM_j).\end{equation} The minimum cost is computed by taking the optimal measurement that minimizes this average \begin{equation}
\tilde{C}_{\mathrm{min}}=\min_{\{M\}}\tilde{C}(M).\end{equation} Minimum-error discrimination corresponds to the case where $C_{ij}=1-\delta_{ij}.$ For a sequence of quantum states, there is a local minimum-cost problem for each component state as well as a global minimum-cost problem associated with the sequence as a whole. The authors in \cite{wallden2014} considered sequences of i.i.d. quantum states where the local cost matrices are identical for the different component states. They proved that under the assumption that the global cost matrix is a linear function of the local cost matrices, the minimum-cost measurement for the sequence is given by the tensor product of the local minimum-cost measurements. In this paper, in contrast, we show that neither the assumption of identical distribution for the component states nor an assumption on the relationship between local and global cost matrices is necessary for proving the optimality of local measurements; the mere independence of the component states is sufficient to guarantee that.

The optimality of fixed local measurements for unambiguous discrimination of quantum sequences were proved in \cite{gupta2024unambiguous} with a number of additional assumptions. There, the states of the sequence were drawn with equal probability from a set of linearly independent pure quantum states whose pairwise inner products were real and equal. Thus the component states of the sequences were identically distributed and also satisfied highly symmetric conditions due to the constraints on their prior probabilities and mutual overlaps. In this paper, we show that one does not need these additional assumptions for the optimality of local measurements; their independence is sufficient to guarantee that.

In the following sections, we prove Theorem \ref{thm:main} for the minimum-error and unambiguous paradigm.

\section{Minimum-error discrimination of quantum sequences}
We need the following two results to prove Theorem \ref{thm:main} for this case. The first one is the Holevo–Yuen–Kennedy–Lax Theorem \cite{yuen1975optimum,watrous2018theory}, which provides a necessary and sufficient condition for the optimality of a minimum-error discrimination measurement for a given ensemble (note that the optimal measurement is not unique \cite{bae2013structure}). We only state the theorem here; the proof can be found in \cite{watrous2018theory}.
\begin{theorem}[Holevo–Yuen–Kennedy–Lax~\cite{yuen1975optimum,watrous2018theory}]\label{thm:me-opt}
    Let $N$ be a positive integer and $\X$ be a complex Euclidean space of finite dimension. Given an ensemble $\E=\{(q_i,\sigma_i):i\in[N]\}$ of density operators on $\X$, a measurement $\{M_1, \ldots, M_N\} $ is optimal for minimum-error discrimination of the elements of $\E$ if and only if 
    \begin{equation} \label{opt-cond}
        \sum_{i=1}^N q_i \sigma_i M_i \succeq q_j \sigma_j
    \end{equation} 
    for all $j \in [N]$.
\end{theorem}
The second result we need is the following lemma.
\begin{lemma}\label{lem:pos-matrices-eigs}
    Let $A, B, C,$ and $D$ be positive semidefinite operators such that $A \succeq C$ and $B \succeq D$. Then, $A \otimes B \succeq C \otimes D$.
\end{lemma}
\begin{proof}
    Given that $A \succeq C$ and $B \succeq D$, we have that $A - C$ and $B - D$ are positive semidefinite. To show $A \otimes B \succeq C \otimes D$, consider the expression
    \begin{equation}
        \begin{aligned}
            A \otimes B - C \otimes D &= 
        (A \otimes B) - (A \otimes D) + (A \otimes D) - (C \otimes D) \\        
        &= (A \otimes (B - D)) + ((A - C) \otimes D).
        \end{aligned}
    \end{equation}
    Since $A - C$ and $B - D$ are positive semidefinite, their tensor products with any positive semidefinite operator also yield positive semidefinite operators. Hence, $A \otimes (B - D)$ and $(A - C) \otimes D$ are positive semidefinite. The sum of two positive semidefinite operators remains positive semidefinite. Therefore, $A \otimes B - C \otimes D$ is positive semidefinite, and thus $A \otimes B \succeq C \otimes D$.
\end{proof}
We now present the proof of Theorem \ref{thm:main} for minimum-error discrimination.
\begin{proof}
    Without loss of generality, assume that an optimal measurement for minimum-error discrimination between the elements of the ensemble $\E^i$ is $M^i=\{M_x^i\}_{x=1}^{\l_i}, \text{ where } M^i_x \in \Pos(\C^{d_i}) \text{ and } \sum_{x=1}^{\l_i}M_x^i=\I_{d_i}$
    for all $x \in [\l_i]$ and $i \in [k]$, where $\Pos(\C^{d_i})$ denotes the set of positive semidefinite operators on $\C^{d_i}$. It follows that 
    \begin{equation} \label{opt-conds}
        \sum_{x=1}^{\l_i}\eta^i_x \rho^i_x M^i_x \succeq \eta^i_y \rho^i_y
    \end{equation}
    for all $y \in [\l_i]$ and $i \in [k]$. This represents a set of $\l_i$ conditions, one for each state in $\E^i$. We select $k$ of these conditions, one for each $i \in [k]$. Our selection can be denoted by a $k$-tuple $(y_1, \dots, y_k)$, which means that we are considering inequality \eqref{opt-conds} for the states $\rho^1_{y_1}, \dots, \rho^k_{y_k}$. 
    
    Taking the tensor product of the operators on both sides of these inequalities and applying Lemma \ref{lem:pos-matrices-eigs} we get 
    \begin{equation} \label{gen-seq-opt}
        \begin{aligned}
            \bigotimes_{i=1}^k \sum_{x_i=1}^{\l_i}\eta^i_{x_i} \rho^i_{x_i} M^i_{x_i} \succeq
            \bigotimes_{i=1}^k \eta^i_{y_i} \rho^i_{y_i},
        \end{aligned}
    \end{equation}
    where $y_i \in [\l_i]$ for all $i \in [k]$. The right side of the above inequality is the sequence $\rho^1_{y_1}\otimes \dots \otimes \rho^k_{y_k}$. By selecting all possible $k$-tuples from inequality \eqref{opt-conds} we will have a collection of conditions, one for each sequence, which is precisely what inequality \eqref{gen-seq-opt} represents. The left-hand side of \eqref{gen-seq-opt} can be expanded as  
    $$\sum_{x_1,\ldots,x_k}\eta^1_{x_1}\cdots \eta^k_{x_k}(\rho^1_{x_1} \otimes \cdots \otimes \rho^k_{x_k})(M^1_{x_1} \otimes \cdots \otimes M^k_{x_k}).$$ 
    As $M^i=\{M_x^i\}_{x=1}^{\l_i} $ forms a measurement on the ensemble $\E^i$, we observe that
   \begin{equation}
        \begin{aligned}
            \sum_{x_1,\ldots,x_k}M^1_{x_1} \otimes \cdots \otimes M^k_{x_k}=&\sum_{x_1=1}^{\l_1}M^1_{x_1} \otimes \cdots \otimes \sum_{x_k=1}^{\l_k}M^k_{x_k}\\ =&\I_{d_1}\otimes \cdots \otimes \I_{d_k}
        \end{aligned}     
    \end{equation}
showing that $\{M^1_{x_1} \otimes \cdots \otimes M^k_{x_k}\}$ is a bona fide measurement on $\E_k$. This, together with Eq. \eqref{gen-seq-opt}, demonstrates that the measurement whose elements are ${M^1_{x_1} \otimes \cdots \otimes M^k_{x_k}}$ is indeed an optimal measurement for discriminating the elements of the ensemble $\E_k$, as established by Theorem \ref{thm:me-opt}.  
    
Therefore, the optimal measurement for discriminating the sequences can be achieved by performing the optimal measurement for each ensemble on its corresponding state. As a result, the probability of correctly identifying the entire sequence is the product of the probabilities of correctly identifying each individual state within the sequence.
\end{proof}

\section{Unambiguous discrimination of quantum sequences} Unlike the minimum-error case, for nontrivial unambiguous discrimination, a set of quantum states has to satisfy a certain condition. In Sec. \ref{sec:unamb_con}, we derive the condition under which a set of quantum sequences can be unambiguously discriminated. We then present the proof of Theorem \ref{thm:main} in Sec. \ref{sec:unamb_thm}, which gives an explicit formula of $p(\E_k)$ in terms of $\{p(\E^i)\}$.

\subsection{Conditions for unambiguous discrimination of quantum sequences} \label{sec:unamb_con}
Let $Q=\{\sigma_1, \ldots, \sigma_N\}$ be a set of quantum states and $\supp(Q)$ denote the Hilbert space spanned by the eigenvectors of the matrices $\{\sigma_1, \ldots, \sigma_N\}$ that correspond to nonzero eigenvalues. Additionally, let $S(\E)$ represent the set of states of an ensemble $\E$.

\begin{lemma}(from Ref. \cite{feng2004unambiguous}) \label{lem:mixed-ud}
    The set of density operators $Q=\{\sigma_1, \ldots, \sigma_N\}$ can be unambiguously discriminated if and only if $\supp(Q) \neq \supp(Q_i)$, where $Q_i=Q\setminus\{\sigma_i\}$ for all $i \in [N]$. If $\sigma_i$ are rank-$1$ operators (pure states), then this condition is same as the set $Q$ being linearly independent. 
\end{lemma}

Since we now want to discriminate the set of sequences unambiguously, they must satisfy Lemma \ref{lem:mixed-ud}. That is, for $S(\E_k)$ to be unambiguously discriminable, it must hold that  
\begin{equation} \label{eq:seq-ud-cond}
    \supp(S(\E_k)) \neq S(\E_k)\setminus\{\rho^1_{x_1}\otimes \cdots \otimes \rho^k_{x_k}\}
\end{equation}  
for all $x_i \in [\l_i]$ and $i\in [k]$. In the following lemma, we show that a set of quantum sequences satisfies this condition if and only if the individual ensembles do. 
\begin{lemma}\label{lem:ud-cond}
    The set of quantum sequences satisfies the condition $ \supp(S(\E_k)) \neq S(\E_k)\setminus\{\rho^1_{x_1}\otimes \cdots \otimes \rho^k_{x_k}\}$ for all $x_i \in [\l_i]$ and $i \in [k]$ if and only if the ensembles $\E^i$ satisfy $\supp(S(\E^i)) \neq \supp(S(\E^i))\setminus \{\rho^i_j\}$ for all $i \in [k]$ and $j \in [\l_i]$.
\end{lemma}
We will need the following two lemmas to prove Lemma \ref{lem:ud-cond}.
\begin{lemma} \label{lem:lemma1}
    Let $k$ be a positive integer and $\V$ be a vector space. For each $i \in [k]$, let $\V_i \subseteq \V$ and $\W_i \subseteq \V_i$ be subspaces. If $a_i \in \V_i$ and $a_i \notin \W_i$, for each $i \in [k]$, then $a_1 \otimes \cdots \otimes a_k \notin \W_1 \otimes \cdots \otimes \W_k$.
\end{lemma}

\begin{proof}
For each $i \in [k]$, let $\{u_1^i, \ldots, u_{s_i}^i\}$ be a basis of $\W_i$, and extend it to a basis $\{u_1^i, \ldots, u_{s_i}^i, v_1^i, \ldots, v_{t_i}^i\}$ of $\V_i$. First, express each $a_i$ in terms of the basis of $\V_i$,
\begin{equation}
    a_i = b_i + \sum_{j=1}^{t_i} \lambda_j^i v_j^i,
\end{equation}
where $b_i \in \W_i$. Since $a_i \notin \W_i$, there exists at least one $j_i$ such that $\lambda_{j_i}^i \neq 0$. Consider the tensor product $a_1 \otimes \cdots \otimes a_k$. Its expansion includes the term
\begin{equation}
    (\lambda_{j_1}^1 \cdots \lambda_{j_k}^k) (v_{j_1}^1 \otimes \cdots \otimes v_{j_k}^k).
\end{equation}
This term is nonzero because each $\lambda_{j_i}^i \neq 0$. Observe that $v_{j_1}^1 \otimes \cdots \otimes v_{j_k}^k \notin \W_1 \otimes \cdots \otimes \W_k$, as each $v_{j_i}^i \notin \W_i$. Finally, note that $a_1 \otimes \cdots \otimes a_k \notin \W_1 \otimes \cdots \otimes \W_k$ because its expansion contains a nonzero term that is not in $\W_1 \otimes \cdots \otimes \W_k$.
\end{proof}

\begin{lemma}\label{lem:lemma2}
    Let $\V_i$ be vector spaces and $B_i = \{b_1^i, \ldots, b_{\ell_i}^i\} \subseteq \V_i$ for $i \in [k]$. Let $\W_i = \spn(B_i)$ be the subspaces spanned by $B_i$. Then, 
    \begin{equation}
        \W_1 \otimes \cdots \otimes \W_k = \spn \{b_{x_1}^1 \otimes \cdots \otimes b_{x_k}^k \mid x_i \in [\ell_i], \ \forall i \in [k] \}.
    \end{equation}
\end{lemma}

\begin{proof}
We will show that $\W_1 \otimes \cdots \otimes \W_k = \spn \{b_{x_1}^1 \otimes \cdots \otimes b_{x_k}^k \mid x_i \in [\ell_i], \ \forall i \in [k] \}$ by proving both inclusions. First, consider 
\begin{equation}
    \spn \{b_{x_1}^1 \otimes \cdots \otimes b_{x_k}^k \mid x_i \in [\ell_i], \ \forall i \in [k] \} \subseteq \W_1 \otimes \cdots \otimes \W_k.
\end{equation}
For all $i \in [k]$ and each $x_i \in [\ell_i]$, we have $b_{x_i}^i \in \W_i$. Therefore, $b_{x_1}^1 \otimes \cdots \otimes b_{x_k}^k \in \W_1 \otimes \cdots \otimes \W_k$. As $\W_1 \otimes \cdots \otimes \W_k$ is a subspace, it contains the span of these tensors.

Now, for the reverse inclusion 
\begin{equation}
    \W_1 \otimes \cdots \otimes \W_k \subseteq \spn \{b_{x_1}^1 \otimes \cdots \otimes b_{x_k}^k \mid x_i \in [\ell_i], \ \forall i \in [k] \}.
\end{equation}
Let $w_1 \otimes \cdots \otimes w_k$ be an arbitrary pure tensor in $\W_1 \otimes \cdots \otimes \W_k$.
For each $i \in [k]$, we can write $w_i = \sum_{j=1}^{\ell_i} \lambda_j^i b_j^i$ as $w_i \in \W_i = \spn(B_i)$.
Expanding the pure tensor
\begin{equation}
w_1 \otimes \cdots \otimes w_k = \sum_{j_1=1}^{\ell_1} \cdots \sum_{j_k=1}^{\ell_k} (\lambda_{j_1}^1 \cdots \lambda_{j_k}^k) (b_{j_1}^1 \otimes \cdots \otimes b_{j_k}^k)
\end{equation}
This expansion shows that $w_1 \otimes \cdots \otimes w_k \in \spn \{b_{x_1}^1 \otimes \cdots \otimes b_{x_k}^k \mid x_i \in [\ell_i], \ \forall i \in [k] \}$.
As $\W_1 \otimes \cdots \otimes \W_k$ is spanned by such pure tensors, the inclusion follows. Therefore, $\W_1 \otimes \cdots \otimes \W_k = \spn \{b_{x_1}^1 \otimes \cdots \otimes b_{x_k}^k \mid x_i \in [\ell_i], \ \forall i \in [k] \}$.
\end{proof}
We are now ready for the proof of Lemma \ref{lem:ud-cond}.
\begin{proof}[Proof of Lemma \ref{lem:ud-cond}]
   We first prove that if $S(\E^i)$ satisfies Lemma \ref{lem:mixed-ud} 
   for each $i$, then $S(\E_{k})$ also satisfies Lemma \ref{lem:mixed-ud}.
   Therefore, we assume that for each $i \in [k]$, $S(\E^i)$ satisfies Lemma \ref{lem:mixed-ud}.
   For each $\rho_j^i \in S(\E^i)$, let $A(i,j) = \{\ket{\psi(i,j,t)}\}_{t=1}^{m_{ij}}$ be the set of eigenvectors corresponding to nonzero eigenvalues of $\rho_j^i$, where $m_{ij}$ is the number of such eigenvectors. Thus, $\supp(\rho_j^i) = \spn(A(i,j))$. Define
    \begin{equation}
        \begin{aligned}
            A^i &= \{\ket{\psi(i,j,t)} : t \in [m_{ij}], \ j \in [\ell_i]\}, \\
            A^i_y &= \{\ket{\psi(i,j,t)} : t \in [m_{ij}], \ j \in [\ell_i], \ j \neq y\} \text{ for } y \in [\ell_i].
        \end{aligned}
    \end{equation}    
    By our assumption, for each $i \in [k]$, $\spn(A^i) \neq \spn(A^i_y)$ for all $y \in [\ell_i]$. Thus, there exists $t_{i,y} \in [m_{iy}]$ such that $\ket{\psi(i,y,t_{i,y})} \notin \spn(A^i_y)$. Suppose, for contradiction, that $S(\E_{k})$ does not satisfy Lemma \ref{lem:mixed-ud}. 
    Then there exist $y_i \in [\ell_i]$ for all $i \in [k]$ such that
    \begin{equation}
        \supp(S(\E_{k})) = \supp(S(\E_{k}) \setminus \{\rho_{y_1}^1 \otimes \cdots \otimes \rho_{y_k}^k\}).    
    \end{equation}
    The eigenvectors corresponding to nonzero eigenvalues of $\rho_{y_1}^1 \otimes \cdots \otimes \rho_{y_k}^k$ are
    \begin{equation}
        \{\ket{\psi(1,y_1,s_{1,y_1})} \otimes \cdots \otimes \ket{\psi(k,y_k,s_{k,y_k})} : s_{i,y_i} \in [m_{i,y_i}]\}.    
    \end{equation}
    Thus,
    \begin{equation}
        \begin{aligned}
        &\supp(S(\E_{k}) \setminus \{\rho_{y_1}^1 \otimes \cdots \otimes \rho_{y_k}^k\}) \\
        &= \spn(\{\ket{\psi(1,x_1,s_{1,x_1})} \otimes \cdots \otimes \ket{\psi(k,x_k,s_{k,x_k})} : s_{i,x_i} \in [m_{i,x_i}], \ x_i \neq y_i \; \forall i \in [k]\}) \\
        &= \bigotimes_{i=1}^k \spn(\{\ket{\psi(i,x_i,s_{i,x_i})}\}), \; s_{i,x_i} \in [m_{i,x_i}], \ x_i \neq y_i \; \forall i \in [k] \\
        &= \spn(A^1_{y_1}) \otimes \cdots \otimes \spn(A^k_{y_k}),
        \end{aligned}    
    \end{equation} where we have used Lemma \ref{lem:lemma1} in the second equality. This equals $\supp(S(\E_{k}))$ by our assumption. However by Lemma \ref{lem:lemma2} we see that, 
    \begin{equation}
        \ket{\psi(1,y_1,t_{1,y_1})} \otimes \cdots \otimes \ket{\psi(k,y_k,t_{k,y_k})} \notin \spn(A^1_{y_1}) \otimes \cdots \otimes \spn(A^k_{y_k}),    
    \end{equation}
    leading to a contradiction.
    
    We now prove the converse by showing that if for some $i \in [k], ~S(\E^i)$ does not satisfy Lemma \ref{lem:mixed-ud}
    then $S(\E_{k})$ also fail to satisfy Lemma \ref{lem:mixed-ud}. 
    Therefore, assume that for some $i \in [k]$, the set $S(\E^i)$ does not satisfy Lemma \ref{lem:mixed-ud}. 
    Then there exists a $y \in [\ell_i]$ such that $\spn(A^i) = \spn(A^i_y)$. This implies $\ket{\psi(i,y,s_{i,y})} \in \spn(A^i_y)$ for all $s_{i,y} \in [m_{i,y}]$. Consider 
    \begin{equation}
        S(\E_{k}) \setminus \{\rho_1^1 \otimes \cdots \otimes \rho_y^i \otimes \cdots \otimes \rho_1^k\}.
    \end{equation} 
    The support of the removed state is
    \begin{equation}
        \spn\{\ket{\psi(1,1,s_{1,1})} \otimes \cdots \otimes \ket{\psi(i,y,s_{i,y})} \otimes \cdots \otimes \ket{\psi(k,1,s_{k,1})} : s_{a,b} \in [m_{a,b}]\}.
    \end{equation}
    Each element 
    \begin{equation}
        \ket{\psi(1,1,s_{1,1})} \otimes \cdots \otimes \ket{\psi(i,y,s_{i,y})} \otimes \cdots \otimes \ket{\psi(k,1,s_{k,1})}    
    \end{equation} 
    is in $\spn(A^1_{x_1}) \otimes \cdots \otimes \spn(A^i_y) \otimes \cdots \otimes \spn(A^k_{x_k})$ where $x_j \neq 1$ for all $j \in [k]$ except for $j=i$, for which $x_j = y$. However, 
    \begin{equation}
        \spn(A^1_{x_1}) \otimes \cdots \otimes \spn(A^i_y) \otimes \cdots \otimes \spn(A^k_{x_k}) \subset \supp(S(\E_{k}) \setminus \{\rho_1^1 \otimes \cdots \otimes \rho_y^i \otimes \cdots \otimes \rho_1^k\}).
    \end{equation}    
    This shows that 
    \begin{equation}
        \supp(S(\E_{k}) \setminus \{\rho_1^1 \otimes \cdots \otimes \rho_y^i \otimes \cdots \otimes \rho_1^k\}) = \supp(S(\E_{k})),
    \end{equation}
    proving that $S(\E_{k})$ fails to satisfy Lemma \ref{lem:mixed-ud}.
\end{proof}
The following theorem is immediate from Lemmas~\ref{lem:mixed-ud} and~\ref{lem:ud-cond}. 
\begin{theorem}[Condition for unambiguous discrimination of quantum sequences] \label{thm:ud-eqv}
    The elements of \( S(\E_k) \) can be unambiguously discriminated if and only if the elements of \( S(\E^i) \) can be unambiguously discriminated for each \( i \).

\end{theorem} 

\subsection{Proof of Theorem \ref{thm:main}} \label{sec:unamb_thm}
We now present the proof of Theorem \ref{thm:main} for the unambiguous case. The unambiguous discrimination problem of mixed quantum states can be cast as a semidefinite program (SDP) \cite{eldar2004optimal}. Let us first present this SDP for an ensemble of states $\E=(Q, q)$, where $Q = \{\sigma_i\}_{i = 1}^N$ is a set of density operators acting on $\C^d$ and $q=(q_1, \ldots, q_N)$ is the vector of \textit{a priori} probabilities of states of $Q$. Without loss of generality, we can assume that the eigenvectors of the states of $Q$ that correspond to nonzero eigenvalues span $\C^d$. This is analogous to assuming that for unambiguous discrimination of pure states, we can take the dimension of the Hilbert space they live in to equal the dimension of the space spanned by them. We define $\widetilde{Q_i}$ as the intersection of all kernels $\K_j$ of $\sigma_j$, excluding $\K_i$. This can be expressed mathematically as $\widetilde{Q_i} = \cap_{j = 1,j\neq i}^{N}\K_j$. We also introduce $\Theta_i$, a $d\times r_i$ matrix whose columns form an orthonormal basis for $\widetilde{Q_i}$, where $r_i$ is the dimension of $\widetilde{Q_i}$. With these definitions in place, we can formulate the SDP for determining the optimal probability of unambiguous discrimination of the ensemble $\E$ as follows: \\
    \begin{widetext}
  \begin{equation} \label{eq:unamb-sdp}
      \begin{minipage}{2.6in}
        \centerline{\underline{Primal problem}}\vspace{0mm}
          \begin{align*}
              \text{maximize:}\quad & \sum_{i=1}^{N}q_i \tr(\sigma_i\Theta_i\Delta_i\Theta_i^{\dagger}) \nonumber \\
    \text{subject to:}\quad & \sum_{i=1}^{N} \Theta_i\Delta_i\Theta_i^{\dagger} \preceq \I, \\ 
    & \Delta_i \succeq 0, \quad \forall i \in [N] \nonumber
          \end{align*}
      \end{minipage}
    \hspace*{13mm}
      \begin{minipage}{2.6in} \vspace{-10mm}
        \centerline{\underline{Dual problem}}\vspace{0mm}
     \begin{align}
          \text{minimize:}\quad & \tr(Z) \nonumber \\ 
          \text{subject to:}\quad & \Theta_i^\dagger\left(Z-q_i \sigma_i\right) \Theta_i \succeq 0 \quad \forall i \in [N], \nonumber  \\ 
            & Z \succeq 0 \nonumber 
        \end{align}
      \end{minipage}
    \end{equation}
\end{widetext}
where $\Delta_i$ is an $r_i \times r_i$ matrix for each $i \in [N]$. The optimization variables in this formulation are the $N$ matrices denoted by $\Delta_i$. Our proof strategy will be to use Slater's theorem, which states that if the primal problem is convex and strict feasibility holds, then the duality gap is zero and the primal and dual optimal values are equal \cite{boyd2004convex}. We will first assume primal and dual optimal variables for the ensemble SDPs. Then we will use these variables to construct primal and dual feasible solutions for the sequence SDP. Finally, we will show that these variables make the primal and dual value equal. Applying Slater's theorem to the sequence SDP, we will conclude that these must be the optimal solutions.
\begin{proof} 
    Consider the ensemble \begin{equation}
        \E^i = \{(\eta^i_1,\rho^i_1),\ldots,(\eta^i_{\l_i},\rho^i_{\l_i})\} \end{equation}
    where $\rho_j^i$ are $d_i\times d_i$ density operators and $j \in [\l_i]$. Let $S^i_j=S(\E^i)\setminus\{\rho^i_j\}$ for some $j \in [\l_i]$ and $\widetilde{S^i_j}$ be the intersection of the kernels of all the density matrices of $\E^i$, except for $\rho^i_j$. That is, $\widetilde{S^i_j}=\cap_{t=1,t\neq j}^{\l_i}\K_t^i$, where $\K_t^i$ is the kernel of $\rho_t^i$. Let $\Theta_j^i$ be a $d_i\times r_j^i$ matrix whose columns form an arbitrary orthonormal basis for $\widetilde{S^i_j}$ (which is of dimension $r_j^i$). The optimal probability of unambiguous discrimination of the elements of $\E^i$ is denoted by $p(\E^i)$ and let $\Delta_j^i$ be the $r_j^i\times r_j^i$ matrices that achieve this optimum. In addition to being positive, these matrices satisfy $\sum_{j=1}^{\l_i}\Theta_j^i\Delta_j^i{\Theta_j^i}^\dagger \preceq \I$. (See the Appendix for the SDPs of the individual ensembles and the sequence.) We denote the optimal dual variable for ensemble $\E^i$ by $Z_i$, which is also a positive matrix and satisfies ${\Theta_j^i}^\dagger\left(Z_i-\eta_j^i \rho_j^i\right) \Theta_j^i \succeq 0$. The optimal primal and dual variables satisfy \begin{equation}
         \sum_{j=1}^{\l_i}\eta_j^i\tr \left(\rho_j^i\Theta_j^i\Delta_j^i{\Theta_j^i}^\dagger\right)=\tr(Z_i)=p(\E^i).
    \end{equation}
    Now, we turn to the sequence ensemble given by Eq. \eqref{eq:seqens} and consider the SDP of its unambiguous discrimination. Note that there are $\l$ number of sequence states in this ensemble, and a state is labeled by the index $(x_1,\ldots,x_k)$, which is a $k$-tuple of indices where $x_i \in [\l_i]$ for all $i \in [k]$. Let us denote by $\widetilde{S}(x_1,\ldots,x_k)$ the intersection of the kernels of all the states of $\E_{k}$ except for $\rho^1_{x_1} \otimes \cdots \otimes \rho^k_{x_k}$. That is, $$\widetilde{S}(x_1,\ldots,x_k)=\bigcap_{\substack{(y_1, \ldots, y_k) \in [\ell_1] \times \cdots \times [\ell_k] \\ (y_1, \ldots, y_k) \neq (x_1, \ldots, x_k)}} \K(y_1,\ldots,y_k)$$ where $\K(y_1,\ldots,y_k)$ is the kernel of $\rho^1_{y_1}\otimes \cdots \otimes \rho^k_{y_k}$. This subspace has a decomposition (see Theorem 9 in the Appendix) as follows \begin{align} \label{eq:seq-ker}
    &\widetilde{S}(x_1,\ldots,x_k)=\widetilde{S_{x_1}^1}\otimes\cdots\otimes \widetilde{S_{x_k}^k} \nonumber \\ &=(\cap_{j=1,j\neq x_1}^{\l_1}\K_j^1)\otimes\cdots\otimes(\cap_{j=1,j\neq x_k}^{\l_k}\K_j^k). \end{align} We first give a prescription to construct matrices $\Theta(x_1,...,x_k)$ whose columns form an orthonormal basis of $\widetilde{S}(x_1,\ldots,x_k)$. From Eq. \eqref{eq:seq-ker}, it can be seen that these matrices are of size $\prod_{i=1}^k d_i \times \prod_{i=1}^k r^i_{x_i} (x_i \in [\l_i])$ and a simple procedure to construct them is to take $\Theta(x_1,...,x_k)=\Theta_{x_1}^{1}\otimes \cdots \otimes\Theta_{x_k}^{k}$. For a feasible primal variable $\Delta(x_1,...,x_k)$, we take the tensor product of the optimal primal variables, $\Delta(x_1,...,x_k)=\Delta_{x_1}^{1}\otimes \cdots \otimes\Delta_{x_k}^{k}$. This operator is positive since it is the tensor product of positive operators. Also, the defined operators satisfy (using Lemma \ref{lem:pos-matrices-eigs}) 
    \begin{equation}
        \begin{aligned}
        & \sum_{(x_1,...,x_k)} \Theta(x_1,...,x_k)\Delta(x_1,...,x_k)\Theta(x_1,...,x_k)^\dagger \\
         =&\sum_{(x_1,...,x_k)}(\Theta_{x_1}^1\Delta_{x_1}^1{\Theta_{x_1}^1}^\dagger)\otimes\cdots\otimes(\Theta_{x_k}^k\Delta_{x_k}^k{\Theta_{x_k}^k}^\dagger)\\
         \preceq & \; \I.
         \end{aligned}
    \end{equation}
Now consider the dual variable $Z=Z_1\otimes\cdots\otimes Z_k$, which is positive due to the positivity of the operators $Z_i$ for all $i \in [k]$. We also know that these variables satisfy ${\Theta_{x_i}^i}^\dagger Z_i \Theta_{x_i}^i\succeq {\Theta_{x_i}^i}^\dagger(\eta_{x_i}^i\rho_{x_i}^i)\Theta_{x_i}^i$ for all $x_i \in [\l_i]$ and $i \in [k]$. By taking tensor products for the indices $(x_1,...,x_k)$ we get, (using Lemma \ref{lem:pos-matrices-eigs}) \begin{widetext}
    \begin{equation}
        \Theta(x_1,...,x_k)^\dagger(Z-\eta_{x_1} \cdots \eta_{x_k}\rho_{x_1}\otimes\cdots\otimes\rho_{x_k})\Theta(x_1,...,x_k) \succeq 0.
    \end{equation} \end{widetext}
This shows that the primal and dual variables we introduced are feasible, and they both make the respective objective value equal to $\prod_i p(\E^i)$ by the trace property of the tensor product. Therefore, $p(\E_{k})$ must be equal to $\prod_i p(\E^i)$.
\end{proof}

\section{Conclusions}
Sequence discrimination problems arise in various quantum information processing tasks, where the objective is to determine the state of a sequence of finite length. The main question here is whether an optimal measurement is local or collective, i.e., whether it suffices to measure the individual members of the sequence or not. In this paper, we showed that as long as the members of the sequences are independently drawn from the same ensemble or even from different ensembles, the optimal measurement in both minimum-error and unambiguous discrimination paradigms is local and comprises fixed measurements on the individual components. It follows that if we give up the condition of independent selection of the members of a sequence, the optimal measurement could be either local and fixed, local and adaptive \cite{acin2005multiple}, or collective (partially or wholly) \cite{calsamiglia2010local}.

Our result also shows that independent sequences do not exhibit nonlocality in the sense some other unentangled systems do, as collective measurement is not necessary for optimal extraction of information from such sequences. However, in the context of quantum state exclusion, one may consider them nonlocal, for an entangled measurement is necessary, as demonstrated by the Pusey–Barrett–Rudolph result \cite{pusey2012reality} and its generalizations \cite{crickmore2020unambiguous}.

A significant part of our proof relied on solving the SDP formulation of the problem using the strong duality theorem \cite{boyd2004convex}. While SDP has been successfully used to solve various problems in quantum information theory \cite{skrzypczyk2023semidefinite,mironowicz2024semi}, we believe our approach could help to solve problems in many-body unentangled systems.

Our result could be immediately applied to settings of quantum key distribution protocols \cite{Bennett_2014, bennett1992quantum} or similar ones. In these protocols, Alice sends a sequence of quantum states, selected independently from a known ensemble, to Bob, who measures each state individually. However, if Bob can store the incoming states in quantum memory, he may think about performing a collective measurement on the entire sequence to extract more information. Our result rules out this possibility, as these sequences are i.i.d.

A natural direction for future research involves doing away with the independence assumption. Can there be sequences of dependent quantum states whose optimal discrimination is still achievable by fixed, local measurements? A systematic study involving the set of states, the nature of the joint probability distribution for the states of a sequence, and the nature of optimal measurement to discriminate them will be interesting to investigate.

\begin{acknowledgments}
    The authors would like to thank Snehasish Roy Chowdhury and Jamie Sikora for fruitful discussions.
\end{acknowledgments}

\appendix

\section{Proof of Theorem \ref{thm:subspace-decomp}}
\begin{theorem} \label{thm:subspace-decomp}
    Let $k \in \mathbb{N}$, and for each $i \in [k]$, let $S(\E^i) = \{\rho_1^i, \ldots, \rho_{\ell_i}^i\}$ be a set of $\ell_i$ quantum states which are density operators on $\mathbb{C}^{d_i}$ for some integer $d_i$. Define $S(\E_{k}) = \{\rho_{x_1}^1 \otimes \cdots \otimes \rho_{x_k}^k : x_i \in [\ell_i], \ \forall i \in [k]\}$ as the set of all $k$-length tensor product sequences where the $i$-th component comes from $S(\E^i)$. 
    
    For each $i \in [k]$ and $j \in [\ell_i]$, let $\widetilde{S^i_j}$ be defined as
    \begin{equation}
        \widetilde{S^i_j} = \bigcap_{m \in [\ell_i], m \neq j} \ker(\rho^i_m).   
    \end{equation}
    For $(x_1, \ldots, x_k) \in [\ell_1] \times \cdots \times [\ell_k]$, let $\widetilde{S}(x_1,\ldots,x_k)$ be defined as
    \begin{equation}
        \widetilde{S}(x_1,\ldots,x_k) = \bigcap_{\substack{(y_1, \ldots, y_k) \in [\ell_1] \times \cdots \times [\ell_k] \\ (y_1, \ldots, y_k) \neq (x_1, \ldots, x_k)}} \ker(\rho^1_{y_1} \otimes \cdots \otimes \rho^k_{y_k}).
    \end{equation}
    Then, for all $(x_1, \ldots, x_k) \in [\ell_1] \times \cdots \times [\ell_k]$, it holds that
    \begin{equation}\label{eq:tensor-product-sequence-decomp}
        \widetilde{S}(x_1,\ldots,x_k) = \widetilde{S^1_{x_1}} \otimes \cdots \otimes \widetilde{S^k_{x_k}}.
    \end{equation}
\end{theorem}
\begin{proof}
    We expand both sides of Eq.~\eqref{eq:tensor-product-sequence-decomp} in terms of the kernels. The right side can be written as
    \begin{equation}
        \begin{aligned}
            \widetilde{S^1_{x_1}} \otimes \cdots \otimes \widetilde{S^k_{x_k}} &= 
            \left(\bigcap_{j=1,j\neq x_1}^{\ell_1}\ker(\rho^1_j)\right) \otimes \cdots \otimes \left(\bigcap_{j=1,j\neq x_k}^{\ell_k}\ker(\rho^k_j)\right) \\
            &= \bigcap_{j_i\neq x_i,\forall i}\left(\ker(\rho^1_{j_1}) \otimes \cdots \otimes \ker(\rho^k_{j_k})\right).
        \end{aligned}        
    \end{equation}
    The second equality follows from the fact that if $\{\V^i_j\}_{j=1}^{\ell_i}$ are subspaces of $\V^i$ for $i\in [k]$, then it holds that
    \begin{equation}
        \left(\bigcap_{j=1}^{\ell_1}\V^1_j\right) \otimes \cdots \otimes \left(\bigcap_{j=1}^{\ell_k}\V^k_j\right) = \bigcap_{(j_1,\ldots,j_k)\in [\ell_1] \times \cdots \times [\ell_k]}\V^1_{j_1} \otimes \cdots \otimes \V^k_{j_k}.        
    \end{equation}
    The left side of Eq.~\eqref{eq:tensor-product-sequence-decomp} can be expanded as follows:
    \begin{equation}
        \begin{aligned}
            \widetilde{S}(x_1,\ldots,x_k) &= \bigcap_{\substack{(y_1,\ldots,y_k)=(1,\ldots,1) \\ (y_1,\ldots,y_k) \neq (x_1,\ldots,x_k)}}^{(\ell_1,\ldots,\ell_k)}\ker(\rho^1_{y_1} \otimes \cdots \otimes \rho^k_{y_k}) \\
            &= \bigcap_{\substack{(y_1,\ldots,y_k)=(1,\ldots,1) \\ (y_1,\ldots,y_k) \neq (x_1,\ldots,x_k)}}^{(\ell_1,\ldots,\ell_k)}\left(\ker(\rho^1_{y_1}) \otimes \mathbb{C}^{d_2} \otimes \cdots \otimes \mathbb{C}^{d_k} + \cdots + \mathbb{C}^{d_1} \otimes \cdots \otimes \mathbb{C}^{d_{k-1}} \otimes \ker(\rho^k_{y_k})\right).
        \end{aligned}
    \end{equation}
    The second equality follows from the fact that if $\varphi_i$ is a linear operator on vector space $\V_i$ for $i \in [k]$ and its kernel is denoted by $\ker(\varphi_i)$, then the following relation holds:
    \begin{equation}
        \ker(\varphi_1 \otimes \cdots \otimes \varphi_k) = \sum_{i=1}^k \V_1 \otimes \cdots \otimes \ker(\varphi_i) \otimes \cdots \otimes \V_k.
    \end{equation}
    Now for each $k$-tuple $(y_1,\ldots,y_k) \neq (x_1,\ldots,x_k)$, consider $\ker(\rho^1_{y_1}) \otimes \mathbb{C}^{d_2} \otimes \cdots \otimes \mathbb{C}^{d_k} + \cdots + \mathbb{C}^{d_1} \otimes \cdots \otimes \mathbb{C}^{d_{k-1}} \otimes \ker(\rho^k_{y_k})$. This sum contains $\ker(\rho^1_{y_1}) \otimes \cdots \otimes \ker(\rho^k_{y_k})$ and thus contains $\bigcap \left(\ker(\rho^1_{j_1}) \otimes \cdots \otimes \ker(\rho^k_{j_k})\right)$ where $j_i \neq x_i$ for all $i\in[k]$. Therefore, we conclude that
    \begin{equation}
        \widetilde{S}(x_1,\ldots,x_k) \supseteq \bigcap_{j_i\neq x_i,\forall i}\left(\ker(\rho^1_{j_1}) \otimes \cdots \otimes \ker(\rho^k_{j_k})\right).
    \end{equation}
    
    For the converse, we will show that the following relation holds:
    \begin{equation}\label{eq:converse_inclusion}
        \left(\bigcap_{j_i\neq x_i,\forall i}\left(\ker(\rho^1_{j_1}) \otimes \cdots \otimes \ker(\rho^k_{j_k})\right)\right)^{\perp} \subseteq \left(\widetilde{S}(x_1,\ldots,x_k)\right)^{\perp},
    \end{equation}
    where, for a subspace $\V$, $\V^\perp$ denotes its orthogonal complement. First, we express the right side of the above relation as follows:
    \begin{equation}
        \begin{aligned}
            \left(\widetilde{S}(x_1,\ldots,x_k)\right)^{\perp}
           &= \left(\bigcap_{\substack{(y_1,\ldots,y_k) \\ (y_1,\ldots,y_k) \neq (x_1,\ldots,x_k)}} \ker(\rho^1_{y_1} \otimes \cdots \otimes \rho^k_{y_k})\right)^{\perp} \\
           &= \sum_{\substack{(y_1,\ldots,y_k) \\ (y_1,\ldots,y_k) \neq (x_1,\ldots,x_k)}} \ker(\rho^1_{y_1} \otimes \cdots \otimes \rho^k_{y_k})^{\perp} \\
           &= \sum_{\substack{(y_1,\ldots,y_k) \\ (y_1,\ldots,y_k) \neq (x_1,\ldots,x_k)}} \ima(\rho^1_{y_1}\otimes\cdots\otimes \rho^k_{y_k})
        \end{aligned}    
    \end{equation}
    where $\ima(\rho)$ denotes the image of an operator $\rho$. In the second equality, we have used the fact that the orthogonal complement of the intersection of subspaces is equal to the sum of the orthogonal complement of individual subspaces: $(\V_1 \cap \cdots \cap \V_n)^\perp=\V_1^\perp + \cdots + \V_n^\perp$, where $\V_i$'s are subspaces. This identity also allows us to write the left side of Eq.~\eqref{eq:converse_inclusion} as a sum of subspaces,
    \begin{equation}\label{eq:converse_inclusion_2}
        \left(\bigcap_{j_i\neq x_i,\forall i}\left(\ker(\rho^1_{j_1}) \otimes \cdots \otimes \ker(\rho^k_{j_k})\right)\right)^{\perp}= \sum_{\substack{(j_1,\ldots,j_k)\\ j_i\neq x_i, \forall i}} \left(\ker(\rho^1_{j_1}) \otimes \cdots \otimes \ker(\rho^k_{j_k})\right)^{\perp}.
    \end{equation}
    
    This can be further expanded by noting that for subspaces $\V_1, \ldots, \V_n$, we have
    \begin{equation}
        \begin{aligned}
            (\V_1 \otimes \cdots \otimes \V_n)^\perp = &~ \V_1^\perp \otimes \cdots \otimes \V_n + \cdots + \V_1 \otimes \cdots \otimes \V_n^\perp \\
            + &~ \V_1^\perp \otimes \V_2^\perp \otimes \cdots \otimes \V_n + \cdots + \V_1 \otimes \cdots \V_{n-1}^\perp \otimes \V_n^\perp \\ 
            &~ \vdots \\
            + &~ \V_1^\perp \otimes \cdots \otimes \V_n^\perp.
        \end{aligned}        
    \end{equation}
Each term in the above sum is of the form $\W_1\otimes\cdots\otimes \W_n$, where $\W_i$ is either $\V_i$ or $\V_i^\perp$, and there is at least one index $a \in [n]$ for which $\W_a = \V_a^\perp$. Therefore, each term on the right side of Eq.~\eqref{eq:converse_inclusion_2} is a sum of terms such as $(\mathcal{L}^1_{j_1}\otimes\cdots\otimes \mathcal{L}^k_{j_k})$. That is, 
    \begin{equation}(\ker(\rho^1_{j_1})\otimes\cdots\otimes\ker(\rho^k_{j_k}))^{\perp} = \mathcal{L}^1_{j_1}\otimes\cdots\otimes \mathcal{L}^k_{j_k}    
    \end{equation} where $\mathcal{L}^i_{j_i}=\ker(\rho^i_{j_i})$ or $\mathcal{L}^i_{j_i}=(\ker(\rho^i_{j_i}))^{\perp}$ and there is at least one index $a \in [k]$ for which $\mathcal{L}^a_{j_a}=(\ker(\rho^a_{j_a}))^{\perp}$. 
    
    Let us now consider a pure tensor $u_1\otimes\cdots\otimes u_k\in \mathcal{L}^1_{j_1}\otimes\cdots\otimes\mathcal{L}^k_{j_k}$, where $u_a\in (\ker(\rho^a_{j_a}))^{\perp}=\ima(\rho^a_{j_a})$ for some $a \in [k]$ and $j_a\neq x_a$. This means that $u_a=\rho^a_{j_a}(v^a_{j_a})$ for some $v^a_{j_a}\in \mathbb{C}^{d_a}$. For some other index $b$, where $\mathcal{L}^b_{j_b}=\ker(\rho^b_{j_b})$, $u_b$ can be written as $\sum_{i_b=1}^{\ell_b}\rho^b_{i_b}(v^b_{i_b})$ where $v^b_{i_b} \in \mathbb{C}^{d_b}$ for all $i_b$. This follows from the assumption that $\supp(A^i)$ (the span of all the eigenvectors of the states of $S^i$ that correspond to nonzero eigenvalues) spans $\mathbb{C}^{d_i}$ for all $i$. Therefore, this pure tensor can be written as
    \begin{equation}
        \begin{aligned}
            &~ u_1\otimes\cdots\otimes u_k \\
           =& \sum_{i=1}^{\ell_1}\rho^1_i(v^1_i) \otimes \cdots \otimes \rho^a_{j_a}(v^a_{j_a})\otimes\cdots\otimes \sum_{i=1}^{\ell_k}\rho^k_i(v^k_i)
        \end{aligned}
    \end{equation}
    where the sum appears in those places whose corresponding $\mathcal{L}$ equals the kernel. This can be written as
    \begin{equation}
        \begin{aligned}
            & \sum \; \rho^1_{j_1}(v^1_{j_1})\otimes\cdots\otimes\rho^a_{j_a}(v^a_{j_a})\otimes\cdots\otimes \rho^k_{j_k}(v^k_{j_k})\\
            =& \sum \; (\rho^1_{j_1}\otimes\cdots\otimes \rho^a_{j_a}\otimes \cdots \otimes \rho^k_{j_k})(v^1_{j_1}\otimes\cdots\otimes v^a_{j_a}\otimes\cdots v^k_{j_k})
        \end{aligned}
    \end{equation}
    where the sum is over those indices whose corresponding $\mathcal{L}$ is equal to the kernel.
    
    Therefore, the pure tensor, which was assumed to be an element of $\mathcal{L}^1_{j_1}\otimes\cdots\otimes\mathcal{L}^k_{j_k}$, is shown to belong to
    \begin{equation}
        \sum_{\substack{(y_1,\ldots,y_k) \\ (y_1,\ldots,y_k) \neq (x_1,\ldots,x_k)}} \ima(\rho^1_{y_1}\otimes\cdots\otimes \rho^k_{y_k})
    \end{equation} 
    since there is at least one index $a \in [k]$ that satisfies $j_a\neq x_a$.
    
    Now observe that $\mathcal{L}^1_{j_1}\otimes\cdots\otimes\mathcal{L}^k_{j_k}$ is spanned by pure tensors and, since 
    \begin{equation}
        \left(\bigcap_{j_i\neq x_i,\forall i}\left(\ker(\rho^1_{j_1}) \otimes \cdots \otimes \ker(\rho^k_{j_k})\right)\right)^{\perp}    
    \end{equation}
    is the sum of spaces of the form $\mathcal{L}^1_{j_1}\otimes\cdots\otimes\mathcal{L}^k_{j_k}$, it is spanned by pure tensors as well. The span of these pure tensors form a subspace in 
    \begin{equation}
        \sum_{\substack{(y_1,\ldots,y_k) \\ (y_1,\ldots,y_k) \neq (x_1,\ldots,x_k)}} \ima(\rho^1_{y_1}\otimes\cdots\otimes \rho^k_{y_k})
    \end{equation} 
    and any element of 
    \begin{equation}
        \left(\bigcap_{j_i\neq x_i,\forall i}\left(\ker(\rho^1_{j_1}) \otimes \cdots \otimes \ker(\rho^k_{j_k})\right)\right)^{\perp}
    \end{equation}
    also belongs to this subspace and hence to \begin{equation}
        \sum_{\substack{(y_1,\ldots,y_k) \\ (y_1,\ldots,y_k) \neq (x_1,\ldots,x_k)}} \ima(\rho^1_{y_1}\otimes\cdots\otimes \rho^k_{y_k}) = \left(\widetilde{S}(x_1,\ldots,x_k)\right)^{\perp}.
    \end{equation} 
    Therefore, the converse is established.
\end{proof}

\section{Semidefinite program of unambiguous discrimination of sequences}
The SDPs of the ensembles and the sequences follow straightforwardly from the SDP given in Eq. \eqref{eq:unamb-sdp} in the main text. However, for the reader's convenience, we write them here explicitly.

 The SDP for the optimal probability of unambiguous discrimination of the states of ensemble $\E^i$ is the following: 
\begin{center} 
  \begin{equation} \label{eq:ens-unamb-sdp}
      \begin{minipage}{2.6in}
          \centerline{\underline{Primal problem}}\vspace{-5mm}
          \begin{align*}
              \text{maximize:}\quad & \sum_{j=1}^{\l_i}\eta_j^i\tr(\rho_j^i\Theta_j^i\Delta_j^i{\Theta_j^i}^\dagger) \nonumber \\
    \text{subject to:}\quad & \sum_{j=1}^{\l_i}\Theta_j^i\Delta_j^i\Theta_j^{i \dagger} \preceq \I, \\ 
    & \Delta_j^i \succeq 0, \quad \forall j \in [\l_i] \nonumber
          \end{align*}
      \end{minipage}
    \hspace*{13mm}
      \begin{minipage}{2.6in}
        \centerline{\underline{Dual problem}}\vspace{-5mm}
        \begin{align}
          \text{minimize:}\quad & \tr(Z_i) \nonumber \\ 
          \text{subject to:}\quad & {\Theta_j^i}^\dagger\left(Z_i-\eta_j^i \rho_j^i\right) \Theta_j^i \succeq 0 \quad \forall j \in [\l_i], \nonumber  \\ 
            & Z_i \succeq 0. \nonumber 
        \end{align}
      \end{minipage}
    \end{equation}
\end{center}
Here, $\Theta_j^i$ is an $d_i\times r_j^i$ matrix whose columns form an arbitrary orthonormal basis for $S_j^i$ (of dimension $r_j^i$).

The ensemble $\E_k$ consists of length $k$ sequences of quantum states that are chosen independently, 
\begin{equation*}
    \E_k = \{(\eta^1_{x_1}\cdots \eta^k_{x_k},\rho^1_{x_1} \otimes \dots \otimes \rho^k_{x_k}): ~x_i \in [\l_i] ~\forall i \in [k]\}.
\end{equation*} The SDP for its optimal unambiguous discrimination is the following:\\\\
\centerline{\underline{Primal problem}} \begin{align} \label{eq:seq-unamb-sdp}
\text{maximize:}\quad & \sum_{x_1=1}^{\l_1}\cdots \sum_{x_k=1}^{\l_k}\eta_{x_1}^1 \cdots \eta_{x_k}^k \tr(\rho^1_{x_1} \otimes \cdots \otimes \rho^k_{x_k}\Theta(x_1,...,x_k)\Delta(x_1,...,x_k){\Theta(x_1,...,x_k)}^\dagger) \nonumber \\
\text{subject to:}\quad & \sum_{x_1=1}^{\l_1}\cdots \sum_{x_k=1}^{\l_k}\Theta(x_1,...,x_k)\Delta(x_1,...,x_k){\Theta(x_1,...,x_k)}^\dagger \preceq \I, \\ 
& \Delta(x_1,...,x_k) \succeq 0, \quad \forall x_i \in [\l_i] ~\forall i \in [k]. \nonumber     \end{align} \\

\centerline{\underline{Dual problem}} \begin{align} 
\text{minimize:}\quad & \tr(Z) \nonumber \\
\text{subject to:}\quad & \Theta(x_1,...,x_k)^\dagger (Z-\eta_{x_1}\cdots \eta_{x_k}\rho_{x_1}\otimes\cdots\otimes\rho_{x_k})\Theta(x_1,...,x_k) \succeq 0 ~\forall x_i \in [\l_i] ~\forall i \in [k], \\ 
& Z \succeq 0. \nonumber     
\end{align}
Here, the matrices $\Theta(x_1,...,x_k)$ and $\Delta(x_1,...,x_k)$ are $\l$ in number, with one for each $k$-tuple $(x_1,...,x_k)$.


\begin{thebibliography}{99}

\bibitem{bell1966problem}J. Bell, On the problem of hidden variables in quantum mechanics, \href{https://journals.aps.org/rmp/abstract/10.1103/RevModPhys.38.447}{{Rev. Mod. Phys.} \textbf{38}, 447 (1966).} 


\bibitem{brunner2014bell}N. Brunner, D. Cavalcanti, S. Pironio, V. Scarani, and S. Wehner, Bell nonlocality, \href{https://journals.aps.org/rmp/abstract/10.1103/RevModPhys.86.419}{{Rev. Mod. Phys.} \textbf{86}, 419 (2014).} 

\bibitem{peres1991optimal}A. Peres, and W. Wootters, Optimal detection of quantum information, \href{https://link.aps.org/doi/10.1103/PhysRevLett.66.1119}{{Phys. Rev. Lett}. \textbf{66}, 1119 (1991).}  

\bibitem{bennett1999quantum}
C.~H.~Bennett, D.~P.~DiVincenzo, C.~A.~Fuchs, T.~Mor, E.~Rains, P.~W.~Shor, J.~A.~Smolin, and W.~K.~Wootters, Quantum nonlocality without entanglement, \href{https://doi.org/10.1103/PhysRevA.59.1070}{{Phys. Rev. A} \textbf{59}, 1070 (1999).}

\bibitem{wootters2006distinguishing}
W.~K.~Wootters, Distinguishing unentangled states with an unentangled measurement, \href{https://doi.org/10.1142/S0219749906001724}{{Intl. J. Quantum Inf.} \textbf{04}, 219 (2006).} 

\bibitem{halder2019strong}
S.~Halder, M.~Banik, S.~Agrawal, and S.~Bandyopadhyay, Strong quantum nonlocality without entanglement, \href{https://doi.org/10.1103/PhysRevLett.122.040403}{{Phys. Rev. Lett.} \textbf{122}, 040403 (2019).} 

\bibitem{chitambar2013revisiting}
E.~Chitambar and M.-H.~Hsieh, Revisiting the optimal detection of quantum information, \href{https://doi.org/10.1103/PhysRevA.88.020302}{{Phys. Rev. A} \textbf{88}, 020302(R) (2013).} 

\bibitem{sentis2016quantum}
G.~Sent{\'\i}s, E.~Bagan, J.~Calsamiglia, G.~Chiribella, and R.~Mu\~{n}oz-Tapia, Quantum change point, \href{https://doi.org/10.1103/PhysRevLett.117.150502}{{Phys. Rev. Lett.} \textbf{117}, 150502 (2016).} 

\bibitem{sentis2017exact}
G.~Sent{\'\i}s, J.~Calsamiglia, and R.~Munoz-Tapia, Exact identification of a quantum change point, \href{https://doi.org/10.1103/PhysRevLett.119.140506}{{Phys. Rev. Lett.} \textbf{119}, 140506 (2017).}

\bibitem{mohan2023generalized}Mohan, A., Sikora, J. \& Upadhyay, S. A generalized framework for quantum state discrimination, hybrid algorithms, and the quantum change point problem. \href{https://arxiv.org/abs/2312.04023}{{\em arXiv:2312.04023}.}

\bibitem{higgins2011multiple}
B.~L.~Higgins, A.~C.~Doherty, S.~D.~Bartlett, G.~J.~Pryde, and H.~M.~Wiseman, Multiple-copy state discrimination: Thinking globally, acting locally, \href{https://doi.org/10.1103/PhysRevA.83.052314}{{Phys. Rev. A} \textbf{83}, 052314 (2011).}

\bibitem{cerutti2024optimal}I. Cerutti, and P. Scudo, Optimal individual and collective measurements for nonorthogonal quantum key distribution signals. \href{https://journals.aps.org/pra/abstract/10.1103/PhysRevA.109.032615}{{Phys. Rev. A} \textbf{109}, 032615 (2024).}

\bibitem{bae2015quantum}
J.~Bae and L.-C.~Kwek, Quantum state discrimination and its applications, \href{https://doi.org/10.1088/1751-8113/48/8/083001}{{J. Phys. A: Math. Theor.} \textbf{48}, 083001 (2015).}

\bibitem{barnett2009quantum}
S.~M.~Barnett and S.~Croke, Quantum state discrimination, \href{https://doi.org/10.1364/AOP.1.000238}{{Advanc. Opt. Photon.} \textbf{1}, 238 (2009).}

\bibitem{chefles2000quantum}
A.~Chefles, Quantum state discrimination, \href{https://doi.org/10.1080/00107510010002599}{{Contemporary Physics} \textbf{41}, 401 (2000).}

\bibitem{bergou200411}
J.~A.~Bergou, U.~Herzog, and M.~Hillery, 11 Discrimination of Quantum States in \href{https://link.springer.com/chapter/10.1007/978-3-540-44481-7_11}{{\emph{Quantum State Estimation}} (Springer Berlin, Heidelberg, 2004) pp. 417.}

\bibitem{gupta2024unambiguous}T. Gupta, S. Murshid, and S. Bandyopadhyay, Unambiguous discrimination of sequences of quantum states, \href{https://journals.aps.org/pra/abstract/10.1103/PhysRevA.109.052222}{{Phys. Rev. A} \textbf{109}, 052222 (2024).}

\bibitem{acin2005multiple}
A.~Ac{\'\i}n, E.~Bagan, M.~Baig, Ll.~Masanes, and R.~Mu\~{n}oz-Tapia, Multiple-copy two-state discrimination with individual measurements, \href{https://doi.org/10.1103/PhysRevA.71.032338}{{Phys. Rev. A} \textbf{71}, 032338 (2005).}

\bibitem{calsamiglia2010local}
J.~Calsamiglia, J.~I.~de Vicente, R.~Mu\~{n}oz-Tapia, and E.~Bagan, Local discrimination of mixed states, \href{https://doi.org/10.1103/PhysRevLett.105.080504}{{Phys. Rev. Lett.} \textbf{105}, 080504 (2010).}

\bibitem{chefles1998unambiguous}
A.~Chefles, Unambiguous discrimination between linearly independent quantum states, \href{https://www.sciencedirect.com/science/article/pii/S0375960198000644}{{Phys. Lett. A} \textbf{239}, 339 (1998).}

\bibitem{feng2004unambiguous}
Y.~Feng, R.~Duan, and M.~Ying, Unambiguous discrimination between mixed quantum states, \href{https://doi.org/10.1103/PhysRevA.70.012308}{{Phys. Rev. A} \textbf{70}, 012308 (2004).}

\bibitem{holevo1973statistical}
A.~S.~Holevo, Statistical decision theory for quantum systems, \href{https://www.sciencedirect.com/science/article/pii/0047259X73900286}{{Journal of Multivariate Analysis} \textbf{3}, 337 (1973).}

\bibitem{helstrom1969quantum}
C.~W.~Helstrom, Quantum detection and estimation theory, \href{https://doi.org/10.1007/BF01007479}{{Journal of Statistical Physics} \textbf{1}, 231 (1969).}

\bibitem{ivanovic1987differentiate}
I.~D.~Ivanovic, How to differentiate between non-orthogonal states, \href{https://www.sciencedirect.com/science/article/abs/pii/0375960187902222}{{Phys. Lett. A} \textbf{123}, 257 (1987).}

\bibitem{dieks1988overlap}
D.~Dieks, Overlap and distinguishability of quantum states, \href{https://www.sciencedirect.com/science/article/abs/pii/0375960188908407}{{Phys. Lett. A} \textbf{126}, 303 (1988).}

\bibitem{peres1988differentiate}
A.~Peres, How to differentiate between non-orthogonal states, \href{https://www.sciencedirect.com/science/article/abs/pii/0375960188910341}{{Phys. Lett. A} \textbf{128}, 19 (1988).}

\bibitem{jaeger1995optimal}
G.~Jaeger and A.~Shimony, Optimal distinction between two non-orthogonal quantum states, \href{https://www.sciencedirect.com/science/article/pii/037596019400919G}{{Phys. Lett. A} \textbf{197}, 83 (1995).}

\bibitem{sun2001optimum}
Y.~Sun, M.~Hillery, and J.~A.~Bergou, Optimum unambiguous discrimination between linearly independent nonorthogonal quantum states and its optical realization, \href{https://doi.org/10.1103/PhysRevA.64.022311}{{Phys. Rev. A} \textbf{64}, 022311 (2001).}

\bibitem{eldar2003designing}
Y.~C.~Eldar, A.~Megretski, and G.~C.~Verghese, Designing optimal quantum detectors via semidefinite programming, \href{https://doi.org/10.1109/TIT.2003.809510}{{IEEE Trans. Inf. Theory} \textbf{49}, 1007 (2003).}

\bibitem{eldar2003semidefinite}
Y.~C.~Eldar, A semidefinite programming approach to optimal unambiguous discrimination of quantum states, \href{https://doi.org/10.1109/TIT.2002.807291}{{IEEE Trans. Inf. Theory} \textbf{49}, 446 (2003).}

\bibitem{wallden2014}
P.~Wallden, V.~Dunjko, and E.~Andersson, Minimum-cost quantum measurements for quantum information, \href{https://iopscience.iop.org/article/10.1088/1751-8113/47/12/125303}{J. Phys. A: Math. Theor. \textbf{47}, 125303 (2014).}

\bibitem{yuen1975optimum}
H.~Yuen, R.~Kennedy, and M.~Lax, Optimum testing of multiple hypotheses in quantum detection theory, \href{https://ieeexplore.ieee.org/document/1055351}{{IEEE Trans. Inf. Theory} \textbf{21}, 125 (1975).}

\bibitem{watrous2018theory}
J.~Watrous, \emph{The Theory of Quantum Information} \href{https://doi.org/10.1017/9781316848142}{(Cambridge University Press, Cambridge, 2018).}

\bibitem{bae2013structure}
J.~Bae, Structure of minimum-error quantum state discrimination, \href{https://doi.org/10.1088/1367-2630/15/7/073037}{{New J. Phys.} \textbf{15}, 073037 (2013).}

\bibitem{eldar2004optimal}
Y.~C.~Eldar, M.~Stojnic, and B.~Hassibi, Optimal quantum detectors for unambiguous detection of mixed states, \href{https://doi.org/10.1103/PhysRevA.69.062318}{{Phys. Rev. A} \textbf{69}, 062318 (2004).}

\bibitem{boyd2004convex}
S. Boyd, and L. Vandenberghe, \emph{Convex Optimization} \href{https://doi.org/10.1017/CBO9780511804441}{(Cambridge University Press, Cambridge, 2004).}

\bibitem{pusey2012reality}
M. Pusey, J. Barrett, and T. Rudolph, On the reality of the quantum state, \href{https://doi.org/10.1038/nphys2309}{{Nature Physics} \textbf{8}, 475 (2012).}

\bibitem{crickmore2020unambiguous}
J. Crickmore, I. Puthoor, B. Ricketti, S. Croke, M. Hillery, and E. Andersson, Unambiguous quantum state elimination for qubit sequences, \href{https://doi.org/10.1103/PhysRevResearch.2.013256}{{Phys. Rev. Res.} \textbf{2}, 013256 (2020).}

\bibitem{skrzypczyk2023semidefinite}
P.~Skrzypczyk and D.~Cavalcanti, \emph{Semidefinite Programming in Quantum Information Science.} \href{https://doi.org/10.1088/978-0-7503-3343-6}{(IOP Publishing, Bristol, UK, 2023).}

\bibitem{mironowicz2024semi}
P. Mironowicz, Semi-definite programming and quantum information, \href{https://iopscience.iop.org/article/10.1088/1751-8121/ad2b85/meta}{{J. Phys. A: Math. Theor.} \textbf{57}, 163002 (2024).}

\bibitem{Bennett_2014}
C. Bennett, and G. Brassard, Quantum cryptography: Public key distribution and coin tossing, \href{http://dx.doi.org/10.1016/j.tcs.2014.05.025}{{Theoretical Computer Science} \textbf{560}, 7 (2014).}

\bibitem{bennett1992quantum}
C. Bennett, Quantum cryptography using any two nonorthogonal states, \href{https://link.aps.org/doi/10.1103/PhysRevLett.68.3121}{{Phys. Rev. Lett.} \textbf{68}, 3121 (1992).}

\end{thebibliography}
\end{document}